\documentclass[12pt]{article}
\usepackage{verbatim,amssymb,amsthm,amsfonts}
\usepackage[nonamelimits]{amsmath}
\usepackage{fullpage}
\usepackage[dvips,pagebackref,colorlinks=true,linkcolor=blue, citecolor=blue]{hyperref}

\usepackage{pdftricks}
\begin{psinputs}
  \usepackage{pstricks}	
  \usepackage{pst-node}
\end{psinputs}

\configure[pdfgraphic][width=2in,rulesep=1pt]

\newcommand\R{\mathbb R}

\def\B{\{0,1\}}

\providecommand\abs[1]{\lvert#1\rvert}

\providecommand\ip[1]{\langle#1\rangle}

\newtheorem{theorem}{Theorem}

\newtheorem{claim}[theorem]{Claim}

\DeclareMathOperator\pr{\mathrm{Pr}}
\DeclareMathOperator\E{\mathrm{E}}

\def\eps{\varepsilon}

\def\cond{\mid}

\def\calD{\mathcal{D}}

\begin{document}

\title{On extracting common random bits \\ from correlated sources}
\author{Andrej Bogdanov\thanks{Department of Computer Science and Engineering and Institute for Theoretical Science and Communications, Chinese University of Hong Kong. Email: {\tt andrejb@cse.cuhk.edu.hk}. Supported by RGC GRF grant 2150617.} \and Elchanan Mossel\thanks{U.C.
Berkeley and Weizmann Institute. E-mail: {\tt mossel@stat.berkeley.edu}. Supported by NSF Career award DMS-0548249 and Israeli Science Foundation grant 1300/08. Supported by Minerva foundation with funding from the Federal German Ministry for Education and Research}}
\date{}

\maketitle
\begin{abstract}
Suppose Alice and Bob receive strings of unbiased independent but noisy bits from some random source. They wish to use their respective strings to extract a common sequence of random bits with high probability but without communicating. How many such bits can they extract? The trivial strategy of outputting the first $k$ bits yields an agreement probability of $(1 - \eps)^k < 2^{-1.44k\eps}$, where $\eps$ is the amount of noise. We show that no strategy can achieve agreement probability better than $2^{-k\eps/(1 - \eps)}$.

On the other hand, we show that when $k \geq 10 + 2 (1 - \eps) / \eps$, there exists a strategy which achieves an agreement probability of $0.003 (k\eps)^{-1/2} \cdot 2^{-k\eps/(1 - \eps)}$.
\end{abstract}

\section{Introduction}
Let $x$ and $y$ be strings in $\B^n$ generated according to the
following random process. First, each bit $x_i$ of $x$ is chosen
independently at random from $\B$. Then each bit $y_i$ of $y$ is
independently  set to equal $x_i$ with probability $1 - \eps$ and $1 -
x_i$ with probability $\eps$ (the latter possibility indicates that
$x_i$ is corrupted). Suppose that Alice and Bob now want to agree on a common random string
with probability at least, say, $1/2$. One possible protocol is for
both of them to output the first $O(1/\epsilon)$ bits of their
respective inputs. We show that no protocol can do better up to the
constant factor. On the other hand we show that this gain by a constant factor can be
achieved for certain values of the parameters.

This scenario relates to the problem of extracting a unique identification (ID) string from process variations. Several works have proposed hardware-based procedures for extracting a unique, uniformly random identifying string from a digital circuit of a given type~\cite{lim:05,SHO:08,YLHMGZ:09}. It has been proposed that such strings can be used for authentication and secret key generation of low-power devices such as RFIDs~\cite{lim:05, SD:07}.

However, such procedures are prone to noise: Different instantiations of the procedure may produce slightly different answers. Can the agreement probability in any pair of instantiations be improved algorithmically while maintaining the uniform distribution of the ID string? Our work addresses this question when the noise is random and independent across the bits.
We note that in applications, the noise can be handled using other methods, for example by incorporating noise tolerance at the receiver end.

The case where the goal of the two parties is to extract a single bit was studied independently
a number of times. It is known that in this case the optimal protocol is for the two parties
to use the first bit. See~\cite{Yang07} for references and for studying the problem of extracting one bit from two correlated sequences
with different correlation structures.

In~\cite{MO05,MORSS06} a related question is studied: If $m$
parties receive noisy versions of a common random string, where the noise
of each party is independent, what is the strategy for the $m$ parties that maximizes
the probability that the parties agree on a {\em single} random bit of output
without communicating? \cite{MO05} shows that for large $m$ using the majority functions on all
bits is superior to using a single bit and \cite{MORSS06} uses hyper contractive inequalities to show that for large $m$, majority is close to being optimal.

The optimality of the single bit protocol for two parties and extraction of one bit
implies that if the goal of the two parties is to maximize the {\em expected} number
of bits they agree on, given that they output $k$ bits, they cannot do better than
output the first $k$ bits.  However, this analysis leaves open the possibility that
there exist a strategy where the two parties may be able to agree on {\em all} the
bits with probability as large as $1 - \eps$.

We prove that this is not the case: The probability of agreement can be at most
$2^{-k\eps/(1 - \eps)}$. In the trivial strategy, where each party outputs its first
$k$ bits, the probability of agreement is $(1 - \eps)^k$. Figure~\ref{fig:upperbound}
shows the ratio between the number of bits allowed by our upper bound and the
performance of the trivial strategy, for any fixed agreement probability.

\begin{figure}
\label{fig:upperbound}
\begin{center}
\input{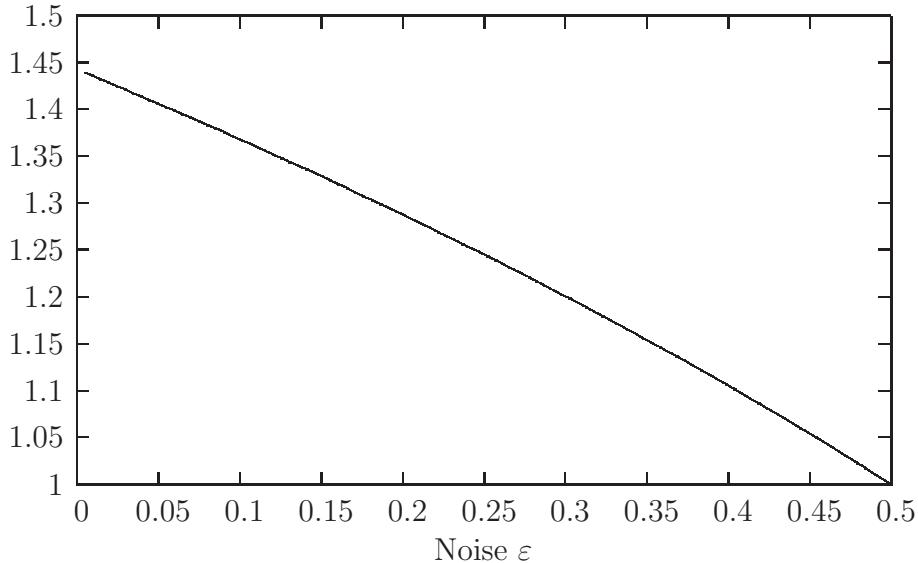}
\hspace{0.5in}
\parbox{0.75\textwidth}{\caption{An upper bound on the factor beyond which the trivial
protocol cannot be outperformed in terms of number of extracted bits (for any probability of agreement).}}
\end{center}
\end{figure}

On the other hand, when the probability of agreement is sufficiently small, an improvement
over the trivial strategy is possible: When $k \geq 10 + 2 (1 - \eps) / \eps$, there exists a protocol which achieves an agreement probability of $0.003 (k\eps)^{-1/2} \cdot 2^{-k\eps/(1 - \eps)}$.

Our protocol is asymptotically almost optimal in the following sense. Suppose we want to achieve a fixed but sufficiently small agreement probability $p$. Our upper bound shows that if the trivial protocol extracts $k$ bits, then no protocol can extract more than $(1/\ln 2)k$ bits. Our protocol can extract $(1/\ln 2 - \delta)k$ bits for any constant $\delta > 0$, as long as $\eps = \eps(\delta)$ is sufficiently small.

G\'acs and K\"orner~\cite{GK} and Witsenhausen~\cite{Wit} show that it is impossible for Alice and Bob to extract $\Omega(n)$ common random bits with probability $1 - o(1)$ for any finite distribution $(x_i, y_i)$, unless $x_i$ and $y_i$ share common randomness. Our work applies to a specific (natural) distribution $(x_i, y_i)$, but yields much sharper bounds. Maurer~\cite{Mau} and Ahlswede and Csisz\`ar~\cite{AC} consider a different model where Alice and Bob can communicate, but eavesdroppers are present and the common random string must remain secret. In this model, it is sometimes possible to achieve better agreement.

\paragraph{Notation} Throughout the paper, we use $n$ to denote the length of the correlated strings $x$ and $y$ available to Alice and Bob, $k$ for the number of bits in their output, and $\eps$ for the noise. The inputs $x = x_1\dots\/x_n$ and $y = y_1\dots\/y_n$, $x_i, y_i \in \B$ are chosen from the following distribution $(x, y)_\eps$: Each pair $x_iy_i$ is independent of all the other pairs and takes the values $00$, $11$ with probability $(1-\eps)/2$ each and the values $01$, $10$ with probability $\eps/2$ each.

\section{The upper bound}
Consider a protocol where Alice and Bob produce $k$ uniform bits of
output. Such a protocol can be described by a pair of functions $f, g:
\B^n \to \B^k$ indicating the outputs produced by Alice and Bob,
respectively.

In our problem, Alice and Bob need to agree on an input that is uniformly
random. We will consider a relaxed scenario where the outputs of Alice and Bob
do not need to be uniformly random, but sufficiently close to having
high ``entropy''. To formalize this we introduce some standard
definitions.

We recall the {\em statistical distance} between $\calD$ and $\calD'$ over sample space $\Omega$ is $\sum_{\omega \in \Omega} \abs{\pr_{\calD}(\omega) - \pr_{\calD'}(\omega)}$.
We say a distribution $\calD$ has {\em min-entropy $t$} the probability
of every element is at most $2^{-t}$. A distribution $\calD$ is {\em
$\delta$-close to min-entropy $t$} if there exists a distribution of
min-entropy $t$ which is within statistical distance
$\delta$ of $\calD$. Abusing notation, we will say that a function $f: \B^n \to
\B^k$ has ($\delta$-close to) min-entropy $t$ if the distribution
$f(x)$, where $x$ is uniform over $\B^n$, has ($\delta$-close to)
min-entropy $t$.

\begin{theorem}
\label{thm:upper}
For any two functions $f, g: \B^n \to \B^k$ that are $\delta$-close to min-entropy $t$ and every $\eps \leq 1/2$,
\[ \pr_{(x, y)_\eps}[f(x) = g(y)] < 2^{-t \epsilon / (1 - \epsilon)} + 2 \delta. \]
\end{theorem}

In particular, if the output of Alice and Bob is exactly uniform then $k = t$ and $\delta = 0$, so if they both output $1/\epsilon$ common bits they cannot hope to agree with probability better than $1/2$.

To prove the theorem, we will use the following two well known claims. Claim~\ref{claim:geometric} follows from the fact that $\E_{(x, y)_\eps}[f(x)g(y)]$ is an inner product of $f$ and $g$. Claim~\ref{claim:boundary} is a corollary of the hypercontractive inequality~\cite{Bon70, Bec75} as it is used in~\cite{KKL88}. The proofs of these claims require some additional notation. We first show how they imply the theorem.

\begin{claim}
\label{claim:geometric}
For every pair of functions $f, g: \B^n \to \R$,
\[ \E_{(x, y)_\eps}[f(x)g(y)] \leq \sqrt{\E_{(x, y)_\eps}[f(x)f(y)]\E_{(x, y)_\eps}[g(x)g(y)]}. \]
\end{claim}

\begin{claim}
\label{claim:boundary}
For every function $h: \B^n \to \B$, $\E[h(x)h(y)] \leq \E[h(x)]^{1 / (1 - \epsilon)}$.
\end{claim}

\begin{proof}[Proof of Theorem~\ref{thm:upper}]
Assume first $f$ and $g$ have min-entropy $t$.
For every $z \in \B^k$, let $f_z: \B^n \to \B$ be the function
\[ f_z(x) = \begin{cases} 1, &\text{if $f(x) = z$} \\ 0, &\text{otherwise}. \end{cases} \]
Define $g_z$ similarly.
Then $\E[f_z(x)]$ and $\E[g_z(x)]$ are upper bounded by $2^{-t}$. Therefore
\begin{align*}
\pr[f(x) = g(y)]
  &= \sum_{z \in \B^k} \pr[f(x) = z \wedge g(y) = z] \\
  &= \sum_{z \in \B^k} \E[f_z(x) g_z(y)] \\
  &\leq \sum_{z \in \B^k} \sqrt{\E[f_z(x) f_z(y)] \cdot \E[g_z(x) g_z(y)]}
    && \text{by Claim~\ref{claim:geometric}} \\
  &\leq \sum_{z \in \B^k} \sqrt{\E[f_z(x)]^{1/(1-\eps)}} \cdot \sqrt{\E[g_z(x)]^{1/(1-\eps)}}
    && \text{by Claim~\ref{claim:boundary}} \\
  &\leq \sqrt{\sum_{z \in \B^k} \E[f_z(x)]^{1/(1-\eps)}} \cdot \sqrt{\sum_{z \in \B^k} \E[g_z(x)]^{1/(1-\eps)}}
     && \text{by Cauchy-Schwarz}
\end{align*}
Since $f$ and $g$ have min-entropy $t$ it follows that $p_z = \E[f_z(x)] \leq 2^{-t}$ and similarly for $g$. We can now bound the expression
in the first square root by
\[
\sum_{z \in \B^k} p_z^{1/(1-\eps)} = \sum_{z \in \B^k} p_z \times p_z^{\eps/(1-\eps)} \leq 2^{-t\eps/(1-\eps)} \sum_{z \in \B^k} p_z = 2^{-t\eps/(1-\eps)}.
\]
By an analogous calculation for the second expression, we obtain that $\pr[f(x) = g(y)] \leq 2^{-t\eps/(1-\eps)}$.

In the case where $f$ and $g$ are $\delta$ close min entropy $t$ distributions we proceed as follows. Let $\delta' > \delta$.
Then by possibly taking a larger value of $n$, there exist $f'$ and $g'$ of min entropy $t$ such that
such that $\pr[f \neq f'] \leq \delta'$ and $\pr[g \neq g'] \leq \delta'$. Now:
\[
\pr[f(x)=g(y)] \leq \pr[f'(x)=g'(y)] + \pr[f(x) \neq f'(x)] + \pr[g(y) \neq g'(y)] \leq 2^{-t\eps/(1-\eps)} + 2 \delta'.
\]
Since $\delta' > \delta$ is arbitrary the proof follows.
\end{proof}
We now prove the two claims. For this we make use of the Fourier expansion of Boolean functions: Every function $f: \B^n \to \R$ can be uniquely written as
\[ f(x) = \sum_{S \subseteq [n]} \hat{f}_S \cdot \chi_S(x) \]
where the {\em character functions} $\chi_S$ are given by
\[ \chi_S(x) = (-1)^{\sum_{i \in S} x_i}. \]
The characters are orthonormal with respect to the inner product $\ip{f,g} = \E[f(x)g(x)]$.

It follows by a calculation that
\begin{equation}
\label{eqn:fourier}
 \E_{(x, y)_\eps}[f(x)g(y)]  = \sum_{S \subseteq [n]} \hat{f}_S \hat{g}_S \rho^{2\abs{S}}
\end{equation}
where $\rho = \sqrt{1 - 2\epsilon}$.

Therefore, to prove Claim~\ref{claim:geometric} we observe that
\begin{align*}
\E_{(x, y)_\eps}[f(x)g(y)]
  &= \sum_{S \subseteq [n]} (\hat{f}_S \rho^{\abs{S}}) \cdot (\hat{g}_S \rho^{\abs{S}}) \\
  &\leq \sqrt{\sum_{S \subseteq [n]} \hat{f}_S^2 \rho^{2\abs{S}}}
      \cdot \sqrt{\sum_{S \subseteq [n]} \hat{g}_S^2 \rho^{2\abs{S}}}
    && \text{by Cauchy-Schwarz} \\
  &= \sqrt{\E_{(x, y)_\eps}[f(x)f(y)] \E_{(x, y)_\eps}[g(x)g(y)]}.
\end{align*}

To prove Claim~\ref{claim:boundary}, we make use of the hypercontractive inequality~\cite{Bon70, Bec75}. This inequality states that for every function $f: \B^n \to \R$, we have
\begin{equation}
\label{eqn:hyper}
\E[((T_\rho f)(x))^2]^{1/2} \leq \E[f(x)^{1+\rho^2}]^{1/(1 + \rho^2)}
\end{equation}
where $T_\rho f: \B^n \to \R$ is defined via the Fourier expansion of $f$ as the function
\[ (T_\rho f)(x) = \sum_{S \subseteq [n]} \hat{f}_S \rho^{\abs{S}} \chi_S(x) \]
Comparing this with (\ref{eqn:fourier}), we have that
\[ \E_x[((T_\rho f)(x))^2] = \E_{(x, y)_\eps}[f(x)f(y)] \]
Where $\rho = \sqrt{1 - 2\epsilon}$.
Now, applying the hypercontractive inequality (\ref{eqn:hyper}) to a function $h: \B^n \to \B$ we obtain
\[ \E_{(x, y)}[h(x)h(y)]^{1/2} \leq \E[h(x)^{1+\rho^2}]^{1/(1 + \rho^2)} = \E[h(x)]^{1/(2 - 2\epsilon)} \]
which proves Claim~\ref{claim:boundary}.

\section{A better strategy}

We now show that when the agreement probability is sufficiently low, the trivial strategy can be outperformed, and in fact one can get strategies that approach the upper bound from Theorem~\ref{thm:upper} to within a constant factor.

\begin{theorem}
\label{thm:lower}
Assume $k \geq 10 + 2 (1 - \eps)/\eps$, and let $n = n(k, \eps)$ be sufficiently large. There exists a function $f\colon \B^n \to \B^k$ such that for all $z \in \B^k$ it holds that
\[
\forall z \in B^k\,\, \pr[f(x) = z] = 2^{-k}, \quad \forall z \in B^k\,\, \pr[f(x) = f(y) = z | f(x) = f(y)] = 2^{-k},
\]
\[
\pr_{(x, y)_\eps}[f(x) = f(y)] \geq 0.003 (\eps k)^{-1/2} 2^{-k\eps/(1 - \eps)}
\]
\end{theorem}

The protocol has the following form. Before starting, Alice and Bob agree on a subset $C$ of $\B^n$ of size $2^k$. On input $x$ (respectively $y$), Alice (respectively Bob) finds and outputs the index of the closest point in $C$ (with an explicit rule in case of ties).
We will show that there exists a choice of $C$ for which (1) each output is generated with the same probability and (2) the probability of agreement is high.

In fact, we prove that on average, a random  subspace of $\B^n$ of dimension $k$ has both properties (1) and (2). In our analysis, we fix $k$ and the noise $\epsilon$ and let $n$ go to infinity.

Let $C$ be an affine subspace of $\B^n$. Write $C=a+L$ where $L$ is a linear subspace.
Let $\prec$ define a strict total order on $\B^n$ with the property that if the Hamming weight of $x$ is smaller than
the Hamming weight of $y$ then $x \prec y$.
We define the regions $R_c, c \in C$ by:
\[ R_c = \{x\colon \text{$x + c  \prec x + c'$ for all $c \neq c' \in C$}\} \]
Note that if $c$ is the unique closest point to $x$ among all the points in $C$ then $x \in R_c$.

Let $H : L \to \B^k$ be any invertible linear map and let $f : \B^n \to \B^k$ be defined as
$f(x) = H(c)$, where $c$ is the unique point such that $x \in R_{a+c}$.

\begin{claim}
\label{claim:same}
For all $z \in \B^k$,
\[ \pr[f(x) = z] = 2^{-k} \quad\text{and}\quad \pr[f(x) = f(y) = z | f(x) = f(y)] = 2^{-k}. \]
\end{claim}

\begin{proof}
If $a+c,a+c' \in C$ and $x + c \in R_{a+c}$ then $x + c' \in R_{a+c'}$. So for every $c \in L$ we have $f(x + c) = f(x) + H(c)$.
Let $z,z' \in \B^k$ and let $z' = z + H(c)$ where $c \in L$. Then $(x,y)$ and $(x+c,y+c)$ have the same distribution and therefore
\[
\pr[f(x) = f(y) = z'] = \pr[f(x+c) = f(y+c) = z + H(c)] = \pr[f(x) = f(y) = z],
\]
and similarly
\[
\pr[f(x) = z'] = \pr[f(x) = z + H(c)] = \pr[f(x+c) = z] = \pr[f(x) = z],
\]
as needed.
\end{proof}

Let $t$ be chosen so that
\[
\frac{1}{\sqrt{2\pi}}\int_t^\infty e^{-z^2/2} dz = 2^{-k-2}
\]
Let $C$ be a random affine space in $\B^n$ of dimension $k$.
Let $r = n/2 + t\sqrt{n}/2$ and note that by the central limit theorem
the hamming ball of radius $r$ contains $2^{n-k-1}(1-o(1))$ points as $n \to \infty$.
We will say $x \in \B^n$ is {\em covered} by $c \in C$ (denoted by $x \in B_c$) if $x$ belongs to the ball of radius $r$ centered at $c$. We say $x$ is {\em uniquely covered} by $c$ (denoted by $x \in U_c$) if it is covered by $c$ but not by any other $c' \in C$. Observe that $U_c \subseteq R_c$.

\begin{claim}
\label{claim:bound}
Let $C$ be a random affine subspace of $\B^n$ of dimension $k$. Then for $n$ sufficiently large,
\[
\E_{C}\pr_{(x, y)_\eps}[\exists c \in C\colon x, y \in U_c] \geq \frac18 \cdot \pr[Z > \sqrt{\eps/(1 - \eps)} t]
\]
where $Z$ is a normal variable of mean $0$ and variance $1$.
\end{claim}

By Claims~\ref{claim:same}~and~\ref{claim:bound}, there must exist a set of points $C$ for which (1) all the regions $R_c$ and of the same size and (2) $\pr_{(x, y)_\eps}[\text{$x, y \in R_c$ for some $c \in C$}] \geq \frac18 \cdot \pr[Z > \sqrt{\eps/(1 - \eps)} t]$. To finish the proof of Theorem~\ref{thm:lower} we calculate a lower bound for the last expression.

\begin{claim}
\label{claim:num_bound}
Let $k \geq 10+2(1-\eps)/\eps$. Then
\[
\frac18 \cdot \pr[Z > \sqrt{\eps/(1 - \eps)} t] \geq 0.003 (\eps k)^{-1/2} 2^{-\eps k/(1-\eps)}
\]
where $Z$ is a normal variable of mean $0$ and variance $1$.
\end{claim}

\begin{proof}
We will use the following estimates valid for every $y > 0$:
\[ \frac{y}{y^2 + 1} e^{-y^2/2} \leq \int_y^\infty e^{-z^2/2} dz \leq \frac{1}{y} e^{-y^2/2}. \]

Note that if $k \geq 10$ then $t \geq 3$ and therefore
\[
2^{-k-2} = \frac{1}{\sqrt{2\pi}}\int_t^\infty e^{-z^2/2} dz \leq \frac{1}{\sqrt{2\pi} t} e^{-t^2/2} \leq e^{-t^2/2}
\]
which implies that $t \leq \sqrt{2 (k+2) \ln 2} \leq \sqrt{2 k}$.
Moreover,
\[
2^{-k} \geq
2^{-k-2} = \frac{1}{\sqrt{2\pi}}\int_t^\infty e^{-z^2/2} dz \geq
\frac{1}{\sqrt{2\pi}} \frac{t}{t^2+1} e^{-t^2/2} \geq \frac{1}{\sqrt{2 \pi} 2t} e^{-t^2/2} \geq e^{-t^2},
\]
which implies that $t \geq \sqrt{k \ln 2}$. So if $k \geq 10+2(1-\eps)/\eps$, then
$t \geq \sqrt{(1-\eps)/\eps}$ and therefore
\begin{align*}
\pr[Z > \sqrt{\eps/(1 - \eps)} t]
&> \frac{1}{\sqrt{2\pi}} \cdot \frac{\sqrt{\eps/(1-\eps)} t}{(\eps / (1 - \eps))t^2 + 1} \cdot e^{-\eps t^2 / 2(1 - \eps)} \\
&\geq \frac1{\sqrt{2\pi}} \cdot \frac{1}{2 \sqrt{\eps/(1-\eps)} t} \cdot \bigl(e^{-t^2 / 2}\bigr)^{\eps/(1 - \eps)} \\
&\geq \frac1{\sqrt{2\pi}} \cdot \frac{1}{2 \sqrt{2 \eps/(1-\eps) k} } 2^{-\eps (k+2)/(1-\eps)} \\
&\geq 0.024 \cdot (\eps k)^{-1/2} 2^{-\eps k/(1-\eps)}
\end{align*}
as needed.
\end{proof}

\begin{proof}[Proof of Claim~\ref{claim:bound}]
Let $C$ be a random affine $k$-dimensional space of $\B^n$.
Such a space can be constructed by starting with a random point $c_0 \sim \B^n$, and iteratively constructing the space $C_i = c_0 + \textrm{span}(c_0 + c_1, \dots, c_0 + c_i)$, where $c_i$ is chosen uniformly from $\B^n\setminus\/C_{i-1}$. Finally let $C = C_k$. From the construction of $C$ it follows that $A C + c$ has the same distribution as $C$ for every invertible linear transformation $A$ and every vector $c$. Since for every pair of vectors $a \neq a', b \neq b'$ there exsits an invertible $A$ and a $c$ such that $A a = b$ and $A a' = b'$ it follows that $\pr_C[a, a' \in C] = \pr_C[b, b' \in C]$. Then
\begin{multline*}
\E_{C}\pr_{(x, y)_\eps}[\exists c \in C\colon x, y \in U_c] \\
\begin{aligned}
&= \E_{C} \sum_{c \in C} \pr_{(x, y)_\eps}[x, y \in U_{c}] \\
&= \E_{C} \sum_{c \in C} \pr_{(x, y)_\eps}[x, y \in B_{c}]\pr_{(x, y)_\eps}[\forall c' \neq c \colon x, y \not\in B_{c'} \mid x, y \in B_{c}] \\
&\geq \E_{C} \sum_{c \in C} \pr_{(x, y)_\eps}[x, y \in B_{c}]\Bigl(1 - \sum_{c' \neq c} \pr_{(x, y)_\eps}[\text{$x \in B_{c'}$ or $y \in B_{c'}$} \mid x, y \in B_{c}] \Bigr) \\
&= \sum_{\B^k} \E_{a \sim \B^n}\Bigl[\pr[x, y \in B_a]\Bigl(1 - \sum_{a' \neq a} \E_{a' \sim \B^n \setminus \{a\}}\pr[\text{$x \in B_{a'}$ or $y \in B_{a'}$} \mid x, y \in B_a] \Bigr)\Bigr].
\end{aligned}
\end{multline*}

The last line uses the fact that the distribution over any pair of points $c \neq c'$ in a random affine space (of dimension at least $1$) is the same as the uniform distribution over pairs $a, a' \in \B^n$ conditioned on $a' \neq a$.
For the expression in the inner summation, we have 
\begin{multline*}
\E_{a' \sim \B^n}\pr_{(x, y)_\eps}[\text{$x \in B_{a'}$ or $y \in B_{a'}$} \mid x, y \in B_a]
\\ \leq 2\E_{a' \sim \B^n}\pr_{(x, y)_\eps}[x \in B_{a'} \mid x, y \in B_a]
= 2\pr_x[x \in B_0] \leq 2^{-k-1}
\end{multline*}
and therefore
\[
\E_{a' \sim \B^n \setminus \{a\}}\pr_{(x, y)_\eps}[x, y \in B_{a'} \mid x, y \in B_a]  \leq 2^{-k-1} \frac{2^n}{2^{n}-1}.
\]
from where the desired expression equals at least
\begin{align*}
\sum_{\B^k}  \E_a\pr_{(x, y)_\eps}[x, y \in B_a] \cdot (1 - (2^k - 1) 2^{-k-1} \frac{2^n}{2^{n}-1})
&> 2^k \cdot \E_a\pr_{(x, y)_\eps}[x, y \in B_a] \cdot (1/2) \\
&= 2^{k-1} \cdot \pr_{(x, y)_\eps}[x, y \in B_0].
\end{align*}
To calculate the last expression, by the two-dimensional central limit theorem we have
\[ \pr_{(x, y)_\eps}[x, y \in B_0] \to \pr_{X, Z}[X > t, \theta X + \sqrt{1 - \theta^2}Z > t] \qquad \text{as} \qquad n \to \infty \]
where $\theta = 1 - 2\eps$ and $X, Z$ are independent normal variables with mean $0$ and variance $1$. We now lower bound this expression:
\begin{align*}
\pr_{X, Z}[X > t, \theta X + \sqrt{1 - \theta^2}Z > t]
&= \pr[X > t]\pr[\theta X + \sqrt{1 - \theta^2}Z > t \cond X > t] \\
&\geq \pr[X > t]\pr[\theta t + \sqrt{1 - \theta^2}Z > t] \\
&= \pr[X > t]\pr[Z > \sqrt{\eps/(1 - \eps)} t].
\end{align*}
Recalling that as $n \to \infty$, $\pr[X > t] \to 2^{-k-2}$, we obtain that as $n$ becomes sufficiently large,
\[ \E_{C}\pr_{(x, y)_\eps}[\text{$x, y \in U_c$ for some $c \in C$}] \geq \frac18 \cdot \pr[Z > \sqrt{\eps/(1 + \eps)} t]. \hfill\qedhere \]
\end{proof}

\section{Conclusion}
In this work we propose the following protocol for two parties that are given access to $n$ noisy random bits with noise of rate $\eps$ to agree on a common random string of length $k$:
\begin{itemize}
\item[] {\bf Preprocessing stage:}
\begin{enumerate}
\item Define a strict total order $\prec$ on $\B^n$ that is consistent with the partial order induced by Hamming weight.
\item Choose a random $k$-dimensional affine subspace $C$ of $\B^n$. Identify the elements of $C$ with strings in $\B^k$.
\end{enumerate}
\item[] {\bf Decoding stage:} On input $x$, output the unique $c \in C$ such that $x + c \prec x + c'$ for all $c' \in C$, $c' \neq c$.
\end{itemize}
Our analysis shows that on average over the choice of $C$, the outputs of Alice and Bob agree with probability $\Omega((k\eps)^{-1/2} 2^{-k\eps/(1-\eps)})$, which is best possible up to a factor of $O(\sqrt{k\eps})$ provided that $k \geq 2/\eps + O(1)$ and $n = n(k, \eps)$ is sufficiently large.

We remark that an explicit upper bound on $n$ in terms of $k$ and $\eps$ can in principle be obtained by using a quantitative version of the central limit theorem in our arguments.

We leave open the question of designing a deterministic and more efficient protocol for the problem considered here. It may also be interesting to investigate how much common randomness can be extracted from other noisy channels $(x, y)$.

\subsection*{Acknowledgments}
We are grateful to Philip Leong for explaining the problem of circuit unique ID extraction which served as the original motivation for the paper and to Uri Feige for crucial insights used in our construction.

\bibliographystyle{alpha}
\bibliography{idrefs}

\newcommand{\etalchar}[1]{$^{#1}$}
\begin{thebibliography}{MOR{\etalchar{+}}06}

\bibitem[AC93]{AC}
R.~Ahlswede and I.~Csisz\`ar.
\newblock Common randomness in information theory and cryptography -- part i:
  Secret sharing.
\newblock {\em IEEE Trans. Inform. Theory}, 39(4):1121--1132, July 1993.

\bibitem[Bec75]{Bec75}
W.~Beckner.
\newblock Inequalities in {Fourier} analysis.
\newblock {\em Annals of Mathematics}, (102):159--182, 1975.

\bibitem[Bon70]{Bon70}
A.~Bonami.
\newblock Etude des coefficients de {Fourier} des fonctions de $l^p(g)$.
\newblock {\em Annales de l'institut Fourier}, 20(2):335--402, 1970.

\bibitem[GK72]{GK}
P.~G\'acs and J.~K\"orner.
\newblock Common information is far less than mutual information.
\newblock {\em Problems of Control and Information Theory}, 2(2):119--162,
  1972.

\bibitem[KKL88]{KKL88}
J.~Kahn, G.~Kalai, and N.~Linial.
\newblock The influence of variables on boolean functions.
\newblock In {\em FOCS '88: Proceedings of the 29th Annual Symposium on
  Foundations of Computer Science}, pages 68--80, Washington, DC, USA, 1988.
  IEEE Computer Society.

\bibitem[LLG{\etalchar{+}}05]{lim:05}
D.~Lim, J.~W. Lee, B.~Gassend, G.~E. Suh, M.~van Dijk, and S.~Devadas.
\newblock {Extracting Secret Keys From Integrated Circuits}.
\newblock {\em IEEE Transactions on Very Large Scale Integration (VLSI)
  Systems}, 13(10):1200--1205, October 2005.

\bibitem[Mau93]{Mau}
U.~Maurer.
\newblock Secret key agreement by public discussion based on common
  information.
\newblock {\em IEEE Trans. Inform. Theory}, 39:733--742, May 1993.

\bibitem[MO05]{MO05}
E.~Mossel and R.~O'Donnell.
\newblock Coin flipping from a cosmic source: On error correction of truly
  random bits.
\newblock {\em Random Structures \& Algorithms}, 4(26):418--436, 2005.

\bibitem[MOR{\etalchar{+}}06]{MORSS06}
E.~Mossel, R.~O'Donnell, O.~Regev, J.~E. Steif, and B.~Sudakov.
\newblock Non-interactive correlation distillation, inhomogeneous {M}arkov
  chains, and the reverse {B}onami-{B}eckner inequality.
\newblock {\em Israel J. Math.}, 154:299--336, 2006.

\bibitem[SD07]{SD:07}
G.~Edward Suh and Srinivas Devadas.
\newblock Physical unclonable functions for device authentication and secret
  key generation.
\newblock In {\em DAC '07: Proceedings of the 44th annual conference on Design
  automation}, pages 9--14, New York, NY, USA, 2007. ACM.

\bibitem[SHO08]{SHO:08}
Ying Su, J.~Holleman, and B.P. Otis.
\newblock A digital 1.6 {pJ/bit} chip identification circuit using process
  variations.
\newblock {\em Solid-State Circuits, IEEE Journal of}, 43(1):69--77, Jan. 2008.

\bibitem[Wit75]{Wit}
H.~S. Witsenhausen.
\newblock On sequences of pairs of dependent random variables.
\newblock {\em SIAM Journal on Applied Mathematics}, 28(1):100--113, 1975.

\bibitem[Yan07]{Yang07}
K.~Yang.
\newblock On the (im)possibility of non-interactive correlation distillation.
\newblock {\em Theoretical Computer Science}, 382(2):157--166, 2007.

\bibitem[YLH{\etalchar{+}}09]{YLHMGZ:09}
Haile Yu, Philip H.~W. Leong, Heiko Hinkelmann, Leandro M{\"o}ller, Manfred
  Glesner, and Peter Zipf.
\newblock Towards a unique fpga-based identification circuit using process
  variations.
\newblock In {\em 19th International Conference on Field Programmable Logic and
  Applications, FPL 2009, August 31 - September 2, 2009, Prague, Czech
  Republic}, pages 397--402, 2009.

\end{thebibliography}
\end{document}